\documentclass[a4paper,final]{scrartcl}

\usepackage[utf8]{inputenc}
\usepackage[T1]{fontenc}

\usepackage{booktabs}
\usepackage{hyperref}
\usepackage{tabularx}
\usepackage{amssymb}
\usepackage{amsthm}
\usepackage{lmodern}

\usepackage{microtype}
\usepackage{amsmath}
\usepackage[nocompress]{cite}

\renewcommand{\geq}{\geqslant}
\renewcommand{\leq}{\leqslant}
\renewcommand{\ge}{\geq}
\renewcommand{\le}{\leq}

\newcommand{\union}{\mathbin{\cup}}

\newcommand{\abs}[1] {\ensuremath\left|#1\right|}
\newcommand{\ceil}[1] {\ensuremath\left\lceil#1\right\rceil}
\newcommand{\set}[2] {\ensuremath{\left\{#1 \mid #2\right\}}}
\newcommand{\os}[1] {\ensuremath{\left\{#1\right\}}}
\newcommand{\N}{\mathbb{N}}
\newcommand{\bigO}{\mathcal{O}}

\newcommand{\FOLL}{\ensuremath{\mathsf{FOLL}}}
\newcommand{\AC}{\ensuremath{\mathsf{AC}}}
\newcommand{\ACC}{\ensuremath{\mathsf{ACC}}}
\newcommand{\qAC}{\ensuremath{\mathsf{qAC}}}
\newcommand{\TC}{\ensuremath{\mathsf{TC}}}
\newcommand{\NC}{\ensuremath{\mathsf{NC}}}
\newcommand{\PTIME}{\ensuremath{\mathsf{P}}}
\newcommand{\LOGSPACE}{\ensuremath{\mathsf{L}}}
\newcommand{\NL}{\ensuremath{\mathsf{NL}}}
\newcommand{\Ppoly}{\ensuremath{\mathsf{P / poly}}}

\newcommand{\vV}{\mathbf{V}}
\newcommand{\vW}{\mathbf{W}}
\newcommand{\vAb}{\mathbf{Ab}}
\newcommand{\vNil}{\mathbf{N}}
\newcommand{\vCom}{\mathbf{Com}}
\newcommand{\vG}{\mathbf{G}}

\newcommand{\CSM}{\ensuremath{\textsc{CSM}}}
\newcommand{\Parity}{\ensuremath{\textsc{Parity}}}
\newcommand{\STCONN}{\ensuremath{\textsc{STConn}}}

\newcommand{\C}{\ensuremath{\mathcal{C}}}

\newcommand{\SCGM}{Cayley groupoid membership}

\newcommand{\SCSM}{Cayley semigroup membership}
\newcommand{\CCSM}{Cayley Semigroup Membership}
\newcommand{\SCGMP}{{\SCGM} problem}

\newcommand{\SCSMP}{{\SCSM} problem}
\newcommand{\CCSMP}{{\CCSM} Problem}
\newcommand{\SLPBP}{logarithmic power basis property}

\newcommand{\SPLCP}{poly-logarithmic circuits property}
\newcommand{\CPLCP}{Poly-Logarithmic Circuits Property}

\newcommand{\ie}{i.e.,~}

\newcommand{\ms}{\hspace*{0.5pt}}

\theoremstyle{plain}
\newtheorem{theorem}{Theorem}
\newtheorem{lemma}[theorem]{Lemma}
\newtheorem{corollary}[theorem]{Corollary}
\newtheorem{proposition}[theorem]{Proposition}
\newtheorem{example}{Example}

\makeatletter
\newcommand{\thickhline}{%
    \noalign {\ifnum 0=`}\fi \hrule height 1pt
    \futurelet \reserved@a \@xhline
}
\makeatother

\title{On the Complexity of the {\CCSMP}}

\author{Lukas Fleischer}
\date{FMI, University of Stuttgart\thanks{This work was supported by the DFG grant DI 435/5--2.}\\
  Universitätsstraße 38, 70569 Stuttgart, Germany\\
  \texttt{fleischer@fmi.uni-stuttgart.de}}

\begin{document}

\maketitle

\begin{abstract}
  We investigate the complexity of deciding, given a multiplication table
  representing a semigroup~$S$, a subset $X$ of $S$ and an element $t$ of $S$,
  whether $t$ can be expressed as a product of elements of $X$.
  It is well-known that this problem is $\NL$-complete and that the more
  general \emph{\SCGMP}, where the multiplication table is not required to be
  associative, is $\PTIME$-complete.
  For groups, the problem can be solved in deterministic log-space which raised
  the question of determining the exact complexity of this variant.
  Barrington, Kadau, Lange and McKenzie showed that for Abelian groups and for
  certain solvable groups, the problem is contained in the complexity class
  $\FOLL$ and they concluded that these variants are not hard for any
  complexity class containing $\Parity$. The more general case of arbitrary
  groups remained open.
  In this work, we show that for both groups and for commutative semigroups,
  the problem is solvable in $\qAC^0$ (quasi-polynomial size circuits of
  constant depth with unbounded fan-in) and conclude that these variants are
  also not hard for any class containing $\Parity$.
  Moreover, we prove that $\NL$-completeness already holds for the classes of
  $0$-simple semigroups and nilpotent semigroups.
  Together with our results on groups and commutative semigroups, we prove the
  existence of a natural class of finite semigroups which generates a variety
  of finite semigroups with $\NL$-complete {\SCSM}, while the {\SCSMP} for the
  class itself is not $\NL$-hard.
  We also discuss applications of our technique to $\FOLL$.
\end{abstract}

\section{Introduction}

The \emph{\SCGMP} (sometimes also called the \emph{generation problem}) asks,
given a multiplication table representing a groupoid $G$, a subset $X$ of $G$
and an element $t$ of $G$, whether $t$ can be expressed as a product of
elements of $X$.
In 1976, Jones and Laaser showed that this problem is
$\PTIME$-complete~\cite{JonesL76}.
Barrington and McKenzie later studied natural subproblems and connected them to
standard subclasses of $\PTIME$~\cite{BarringtonM91}.

When restricting the set of valid inputs to inputs with an associative
multiplication table, the problem becomes $\NL$-complete~\cite{JonesLL76}.
We will call this variant of the problem the \emph{\SCSMP} and analyze its
complexity when further restricting the semigroups encoded by the input.
For a class of finite semigroups $\vV$, the \emph{{\SCSMP} for $\vV$} is
formally defined as follows.
\vspace{1em}

\noindent
\begin{tabularx}{\textwidth}{p{1.15cm}X}
  \thickhline
  $\CSM(\vV)$ \\
  \hline
  \textsf{Input}: & The Cayley table of a semigroup $S \in \vV$, a set $X \subseteq S$ and an element $t \in S$ \\
  \textsf{Question}: & Is $t$ in the subsemigroup of $S$ generated by $X$? \\
  \hline
\end{tabularx}
\vspace{1ex}

The motivation for investigating this problem is two-fold.
Firstly, there is a direct connection between the {\SCSMP} and decision
problems for regular languages: a language $L \subseteq \Sigma^+$ is regular if
and only if there exist a finite semigroup~$S$, a morphism $\varphi \colon
\Sigma^+ \to S$ and a set $P \subseteq S$ such that $L = \varphi^{-1}(P)$.
Thus, morphisms to finite semigroups can be seen as a way of encoding regular
languages. For encoding such a semigroup, specifying the multiplication table
is a natural choice.
Deciding emptiness of a regular language represented by a morphism $\varphi
\colon \Sigma^+ \to S$ to a finite semigroup $S$ and a set $P \subseteq S$
boils down to checking whether any of the elements from the set $P$ is
contained in the subsemigroup of $S$ generated by the images of the letters of
$\Sigma$ under $\varphi$.
Conversely, the {\SCSMP} is a special case of the emptiness problem for regular
languages: an element $t \in S$ is contained in the subsemigroup generated by a
set $X \subseteq S$ if and only if the language $\varphi^{-1}(P)$ with $\varphi
\colon X^+ \to S, x \mapsto x$ and $P = \os{t}$ is non-empty.

Secondly, we hope to get a better understanding of the connection between
algebra and low-level complexity classes included in $\NL$ in a fashion similar
to the results of~\cite{BarringtonM91}.
In the past, several intriguing links between so-called \emph{varieties
of finite semigroups} and the computational complexity of algebraic problems
for such varieties were made.
For example, the fixed membership problem for a regular language was shown to
be in $\AC^0$ if its syntactic monoid is aperiodic, in $\ACC^0$ if the
syntactic monoid is solvable and $\NC^1$-complete
otherwise~\cite{Barrington86,bt88jacm}.
It is remarkable that in most results of this type, both the involved
complexity classes and the algebraic varieties are natural.
On a language-theoretical level, varieties of finite semigroups correspond to
subclasses of the regular languages closed under Boolean operations, quotients
and inverse morphisms.

\subparagraph{Related Work.}

We already mentioned the work of Jones and Laaser on the
{\SCGMP}~\cite{JonesL76}, the work of Jones, Lien and Laaser on the
{\SCSMP}~\cite{JonesLL76} and the work of Barrington and McKenzie on
subproblems thereof~\cite{BarringtonM91}.
The semigroup membership problem and its restrictions to varieties of finite
semigroups was also studied for other encodings of the input, such as matrix
semigroups~\cite{Babai85,BabaiS84,BabaiBIL96} or transformation
semigroups~\cite{Sims1967,FurstHopcroftLuks80,BabaiLS87,Beaudry88,BeaudryMT92,Beaudry94,Beaudry88thesis}.

The group version of the {\SCSMP} ($\CSM(\vG)$, using our notation) was first
investigated by Barrington and McKenzie in 1991~\cite{BarringtonM91}.
They observed that the problem is in symmetric log-space, which has been shown
to be the same as deterministic log-space by Reingold in
2008~\cite{Reingold08}, and suggested it might be complete for deterministic
log-space. However, all attempts to obtain a hardness proof failed (in fact,
their conjecture is shown to be false in this work).
There was no progress in a long time until Barrington, Kadau, Lange and
McKenzie showed that for Abelian groups and certain solvable groups, the
problem lies in the complexity class $\FOLL$ and thus, cannot be hard for any
complexity class containing $\Parity$ in 2001~\cite{BarringtonKLM01}. The case
of arbitrary groups remained open.

\subparagraph{Our Contributions.}

We generalize previous results on Abelian groups to arbitrary commutative
semigroups.
Then, using novel techniques, we show that the {\SCSMP} for the variety of
finite groups $\vG$ is contained in $\qAC^0$ and thus, cannot be hard for any
class containing $\Parity$.
Our approach relies on the existence of succinct representations of group
elements by algebraic circuits.
More precisely, it uses the fact that every element of a group $G$ can be
computed by an algebraic circuit of size $\mathcal{O}(\log^3 \abs{G})$ over any
set of generators.
Since in the {\SCSMP}, the algebraic structure is not fixed, we introduce
so-called Cayley circuits, which are similar to regular algebraic circuits but
expect the finite semigroup to be given as part of the input. We prove that
these Cayley circuits can be simulated by sufficiently small unbounded fan-in
Boolean circuits. We then use this kind of simulation to evaluate all Cayley
circuits, up to a certain size, in parallel.

By means of a closer analysis and an extension of the technique used by Jones,
Lien and Laaser in~\cite{JonesLL76}, we also show that the {\SCSMP} remains
$\NL$-complete when restricting the input to $0$-simple semigroups or to
nilpotent semigroups.

Combining our results, we obtain that the {\SCSMP} for the class $\vG \union
\vCom$, which consists of all finite groups and all finite commutative
semigroups, is decidable in $\qAC^0$ (and thus not $\NL$-hard) while the
{\SCSMP} for the minimal variety of finite semigroups containing $\vG \union
\vCom$ is $\NL$-complete.

Finally, we discuss the extent to which our approach can be used to establish
membership of {\SCSM} variants to the complexity class $\FOLL$.
Here, instead of simulating all circuits in parallel, we use an idea based on
repeated squaring. This technique generalizes some of the main concepts used
in~\cite{BarringtonKLM01}.

\section{Preliminaries}
\label{sec:prelim}

\subparagraph{Algebra.}

A semigroup $T$ is a \emph{subsemigroup} of $S$ if $T$ is a subset of $S$
closed under multiplication.
The \emph{direct product} of two semigroups $S$ and $T$ is the Cartesian
product $S \times T$ equipped with componentwise multiplication.
A subsemigroup of a direct product is also called \emph{subdirect product}.
A semigroup $T$ is a \emph{quotient} of a semigroup $S$ if there exists a
surjective morphism $\varphi \colon S \to T$.

A \emph{variety of finite semigroups} is a class of finite semigroups which is
closed under finite subdirect products and under quotients.
Since we are only interested in finite semigroups, we will henceforth use the
term \emph{variety} for a variety of finite semigroups.
Note that in the literature, such classes of semigroups are often called
\emph{pseudovarieties}, as opposed to Birkhoff varieties which are also closed
under infinite subdirect products.
The following varieties play an important role in this paper:
\begin{itemize}
  \item $\vG$, the class of all finite groups,
  \item $\vAb$, the class of all finite Abelian groups,
  \item $\vCom$, the class of all finite commutative semigroups,
  \item $\vNil$, the class of all finite nilpotent semigroups, \ie{}semigroups where the only idempotent is a zero element.
\end{itemize}
The \emph{join} of two varieties $\vV$ and $\vW$, denoted by $\vV \lor \vW$, is
the smallest variety containing both $\vV$ and $\vW$.
A semigroup $S$ is \emph{$0$-simple} if it contains a zero element $0$ and if
for each $s \in S \setminus \os{0}$, one has $SsS = S$.
The class of finite $0$-simple semigroups does not form a variety.

\subparagraph{Complexity.}

We assume familiarity with standard definitions from circuit complexity.
A function has \emph{quasi-polynomial} growth if it is contained in
$2^{\bigO(\log^c n)}$ for some fixed $c \in \N$.
Throughout the paper, we consider the following unbounded fan-in Boolean
circuit families:
\begin{itemize}
  \item $\AC^0$, languages decidable by circuit families of depth $\bigO(1)$
    and polynomial size,
  \item $\qAC^0$, languages decidable by circuit families of depth $\bigO(1)$
    and quasi-polynomial size,
  \item $\FOLL$, languages decidable by circuit families of depth
    $\bigO(\log \log n)$ and polynomial size,
  \item $\AC^1$, languages decidable by circuit families of depth
    $\bigO(\log n)$ and polynomial size,
  \item $\Ppoly$, languages decidable by circuit families of polynomial size
    (and unbounded depth).
\end{itemize}

We allow NOT gates but do not count them when measuring the depth or the size
of a circuit.
We will also briefly refer to the complexity classes $\ACC^0$, $\TC^0$,
$\NC^1$, $\LOGSPACE$ and $\NL$.

It is known that the \Parity~function cannot be computed by $\AC^0$, $\FOLL$ or
$\qAC^0$ circuits. This follows directly from H{\aa}stad's and Yao's famous
lower bound results~\cite{Hastad86,Yao85}, which state that the number of
Boolean gates required for a depth-$d$ circuit to compute $\Parity$ is
exponential in $n^{1/(d-1)}$.

\section{Hardness Results}

Before looking at parallel algorithms for the {\SCSMP}, we establish two new
$\NL$-hardness results.
To this end, we first analyze the construction already used by Jones, Lien and
Laaser~\cite{JonesLL76}.
It turns out that the semigroups used in their reductions are $0$-simple which
leads to the following result.

\begin{theorem}
  For a class containing all $0$-simple semigroups, the {\SCSMP} is
  $\NL$-complete.
  \label{thm:0-simple}
\end{theorem}

\begin{proof}
  To keep the proof self-contained, we briefly describe the reduction from
  the connectivity problem for directed graphs (henceforth called $\STCONN$) to
  the {\SCSMP} given in~\cite{JonesLL76}.

  Let $G = (V, E)$ be a directed graph.
  We construct a semigroup on the set $S = V \times V \union \os{0}$ where $0$
  is a zero element and the multiplication rule for the remaining elements is
  \begin{equation*}
    (v, w) \cdot (x, y) =
    \begin{cases}
      (v, y) & \text{if $w = x$}, \\
      0 & \text{otherwise}.
    \end{cases}
  \end{equation*}
  By construction, the subsemigroup of $S$ generated by $E \union \set{(v,
  v)}{v \in V}$ contains an element $(s, t)$ if and only if $t$ is reachable
  from $s$ in $G$.
  To see that the semigroup $S$ is $0$-simple, note that for pairs of arbitrary
  elements $(v, w) \in V \times V$ and $(x, y) \in V \times V$, one has $(x, v)
  (v, w) (w, y) = (x, y)$, which implies $S (v, w) S = S$.
\end{proof}

In order to prove $\NL$-completeness for another common class of semigroups, we
need a slightly more advanced construction reminiscent of the ``layer
technique'', which is usually used to show that $\STCONN$ remains
$\NL$-complete when the inputs are acyclic graphs.

\begin{theorem}
  $\CSM(\vNil)$ is $\NL$-complete (under $\AC^0$ many-one reductions).
  \label{thm:nil}
\end{theorem}

\begin{proof}
  Following the proof of Theorem~\ref{thm:0-simple}, we describe an $\AC^0$
  reduction of $\STCONN$ to $\CSM(\vNil)$.

  Let $G = (V, E)$ be a directed graph with $n$ vertices.
  We construct a semigroup on the set $S = V \times \os{1, \dots, n-1} \times V
  \union \os{0}$ where $0$ is a zero element and the multiplication rule for
  the remaining elements is
  \begin{equation*}
    (v, i, w) \cdot (x, j, y) =
    \begin{cases}
      (v, i + j, y) & \text{if $w = x$ and $i + j < n$}, \\
      0 & \text{otherwise}.
    \end{cases}
  \end{equation*}
  The subsemigroup of $S$ generated by $\set{(v, 1, w)}{v = w \text{~or~} (v,
  w) \in E}$ contains an element $(s, n-1, t)$ if and only if $t$ is reachable
  (in less than $n$ steps) from $s$ in $G$.
  Clearly, the zero element is the only idempotent in $S$, so $S$ is nilpotent.
  Also, it is readily verified that the reduction can be performed by an
  $\AC^0$ circuit family.
\end{proof}

\section{Parallel Algorithms for {\CCSM}}

Algebraic circuits can be used as a succinct representation of elements in an
algebraic structure.
This idea will be the basis of the proof that $\CSM(\vG)$ is in $\qAC^0$.
Unlike in usual algebraic circuits, in the context of the {\SCSMP}, the
algebraic structure is not fixed but given as part of the input.
We will introduce so-called Cayley circuits to deal with this setting.
Since these circuits will be used for the {\SCSMP} only, we confine ourselves
to cases where the algebraic structure is a finite semigroup.

\subsection{Cayley Circuits}

A \emph{Cayley circuit} is a directed acyclic graph with topologically ordered
vertices such that each vertex has in-degree $0$ or $2$.
In the following, to avoid technical subtleties when squaring an element, we
allow multi-edges.
The vertices of a Cayley circuit are called \emph{gates}. The vertices with
in-degree $0$ are called \emph{input gates} and vertices with in-degree $2$ are
called \emph{product gates}. Each Cayley circuit also has a designated gate of
out-degree $0$, called the \emph{output gate}. For simplicity, we assume that
the output gate always corresponds to the maximal gate with regard to the
vertex order.
The \emph{size} of a Cayley circuit $\C$, denoted by $\abs{\C}$, is the number
of gates of $\C$.
An \emph{input} to a Cayley circuit $\C$ with $k$ input gates consists of a
finite semigroup $S$ and elements $x_1, \dots, x_k$ of $S$.
Given such an input, the \emph{value} of the $i$-th input gate is $x_i$ and the
value of a product gate, whose predecessors have values $x$ and $y$, is the
product $x \cdot y$ in $S$. The \emph{value of the circuit $\C$} is the value
of its output gate.
We will denote the value of $\C$ under a finite semigroup $S$ and elements
$x_1, \dots, x_k \in S$ by $\C(S, x_1, \dots, x_k)$.

A Cayley circuit can be seen as a circuit in the usual sense: the finite
semigroup $S$ and the input gate values are given as part of the input and the
functions computed by product gates map a tuple, consisting of semigroup $S$
and two elements of $S$, to another element of $S$.
We say that a Cayley circuit with $k$ input gates can be \emph{simulated} by a
family of unbounded fan-in Boolean circuits $(\C_n)_{n \in \N}$ if, given the
encodings of a finite semigroup $S$ and of elements $x_1, \dots, x_k$ of $S$ of
total length $n$, the circuit $\C_n$ computes the encoding of $\C(S, x_1,
\dots, x_k)$.
For a semigroup $S$ with $N$ elements, we assume that the elements of $S$ are
encoded by the integers $\os{0, \dots, N-1}$ such that the encoding of a single
element uses $\ceil{\log N}$ bits. The semigroup itself is given as a
multiplication table with $N^2$ entries of $\ceil{\log N}$ bits each.

\begin{proposition}
  Let $\C$ be a Cayley circuit of size $m$.
  Then, $\C$ can be simulated by a family of unbounded fan-in constant depth
  Boolean circuits $(\C_n)_{n \in \N}$ of size at most $n^m$.
  \label{prop:simulation}
\end{proposition}

\begin{proof}
  Let $\C$ be a Cayley circuit with $k$ input gates and $m-k$ product gates.
  We want to construct a Boolean circuit which can be used for all
  finite semigroups $S$ with a fixed number of elements $N$.
  The input to such a circuit consists of $n = (N^2 + k) \ceil{\log N}$ bits.

  For a fixed vector $(y_1, \dots, y_m) \in S^m$, one can check using a single
  AND~gate (and additional NOT~gates at some of the incoming wires) whether
  $(y_1, \dots, y_m)$ corresponds to the sequence of values occurring at the
  gates of $\C$ under the given inputs.
  To this end, for each gate $i \in \os{1, \dots, m}$ of $\C$, we add
  $\ceil{\log N}$ incoming wires to this AND~gate:
  if the $i$-th gate of $\C$ is an input gate, we feed the bits of the
  corresponding input value into the AND~gate, complementing the $j$-th bit if
  the $j$-th bit of $y_i$ is zero.
  If the $i$-th gate is a product gate and has incoming wires from gates $\ell$
  and $r$, we connect the entry $(y_\ell, y_r)$ of the multiplication table to
  the AND~gate, again complementing bits corresponding to 0-bits of $y_i$.

  To obtain a Boolean circuit simulating $\C$, we put such AND~gates for all
  vectors of the form $(y_1, \dots, y_m) \in S^m$ in parallel. In a second
  layer, we create $\ceil{\log N}$ OR gates and connect the AND~gate for a
  vector $(y_1, \dots, y_m)$ to the $j$-th OR~gate if and only if the $j$-th
  bit of $y_m$ is one.
  The idea is that exactly one of the AND~gates\,---\,the gate corresponding to
  the vector of correct guesses of the gate values of $\C$\,---\,evaluates
  to~$1$ and the corresponding output value $y_m$ then occurs as output value
  of the OR~gates.

  This circuit has depth $2$ and size $N^m + \ceil{\log N} \le n^m$.
\end{proof}

\subsection{The {\CPLCP}}

When analyzing the complexity of $\CSM(\vAb)$, Barrington {et al.} introduced
the so-called \emph{\SLPBP}. A class of semigroups has the {\SLPBP} if any set
of generators $X$ for a semigroup $S$ of cardinality $N$ from the family has
the property that every element of $S$ can be written as a product of at most
$\log(N)$ many powers of elements of $X$.
In~\cite{BarringtonKLM01}, it was shown that the class of Abelian groups has
the {\SLPBP}. Using a different technique, this result can easily be extended to
arbitrary commutative semigroups.

\begin{lemma}
  The variety $\vCom$ has the {\SLPBP}.
  \label{lem:com}
\end{lemma}
\begin{proof}
  Suppose that $S$ is a commutative semigroup of size $N$ and let $X$ be a set
  of generators for $S$.
  Let $y \in S$ be an arbitrary element.
  We choose $k \in \N$ to be the smallest value such that there exist elements
  $x_1, \dots, x_k \in X$ and integers $i_1, \dots, i_k \in \N$ with $y =
  x_1^{i_1} \cdots x_k^{i_k}$.
  Assume, for the sake of contradiction, that $k > \log(N)$.

  The power set $\mathcal{P}(\os{1, \dots, k})$ forms a semigroup when equipped
  with set union as binary operation.
  Consider the morphism $h \colon \mathcal{P}(\os{1, \dots, k}) \to S$ defined
  by $h(\os{j}) = x_j^{i_j}$ for all $j \in \os{1, \dots, k}$. This morphism is
  well-defined because $S$ is commutative.

  Since $\abs{\mathcal{P}(\os{1, \dots, k})} = 2^k > 2^{\log(N)} = \abs{S}$, we
  know by the pigeon hole principle that there exist two sets $K_1, K_2
  \subseteq \os{1, \dots, k}$ with $K_1 \ne K_2$ and $h(K_1) = h(K_2)$. We may
  assume, without loss of generality, that there exists some $j \in K_1
  \setminus K_2$.
  Now, because
  \begin{equation*}
    y = h(\os{1, \dots, k}) = h(K_1) \ms h(\os{1, \dots, k} \setminus K_1) = h(K_2) \ms h(\os{1, \dots, k} \setminus K_1)
  \end{equation*}
  and since neither $K_2$ nor $\os{1, \dots, k} \setminus K_1$ contain $j$, we
  know that $y$ can be written as a product of powers of elements $x_i$ with $1
  \le i \le k$ and $i \ne j$, contradicting the choice of~$k$.
\end{proof}

For the analysis of arbitrary groups, we introduce a more general concept. It
is based on the idea that algebraic circuits (Cayley circuits with fixed
inputs) can be used for succinct representations of semigroup elements.

\begin{example}
  Let $e \in \N$ be a positive integer.
  Then, one can construct a Cayley circuit of size at most $2\ceil{\log e}$
  which computes, given a finite semigroup $S$ and an element $x \in S$ as
  input, the power $x^e$ in $S$.
  If $e = 1$, the circuit only consists of the input gate.
  If $e$ is even, the circuit is obtained by taking the circuit for $e/2$,
  adding a product gate and creating two edges from the output gate of the
  circuit for $e/2$ to the new gate.
  If $e$ is odd, the circuit is obtained by taking the circuit for $e-1$ and
  connecting it to a new product gate. In this case, the second incoming edge
  for the new gate comes from the input gate.
  \label{ex:power}
\end{example}

A class of semigroups has the \emph{\SPLCP} if there exists a constant $c \in
\N$ such that for each semigroup $S$ of cardinality $N$ from the class, for
each subset $X$ of $S$ and for each $y$ in the subsemigroup generated by $X$,
there exists a Cayley circuit $\C$ of size $\log^c(N)$ with $k$ input gates and
there exist $x_1, \dots, x_k \in X$ such that $\C(S, x_1, \dots, x_k) = y$.

\begin{proposition}
  Let $\vV$ be a family of semigroups which is closed under subsemigroups and
  has the {\SLPBP}. Then $\vV$ has the {\SPLCP}.
  \label{prop:lpbp-plrp}
\end{proposition}
\begin{proof}
  Let $X$ be a subset of a semigroup $S$ of cardinality $N$. Let $y$ be in the
  subsemigroup generated by $X$.
  Then, we have $y = x_1^{i_1} \cdots x_k^{i_k}$ for some $x_1, \dots, x_k \in
  X$  with $k \le \log(N)$ and $i_1, \dots, i_k \in \N$.
  By the pigeon hole principle, we may assume without loss of generality that
  $1 \le i_1, \dots, i_k \le N$.
  Using the method from Example~\ref{ex:power}, one can construct Cayley
  circuits $\mathcal{C}_1, \dots, \mathcal{C}_k$ of size at most
  $2\ceil{\log N}$ such that $\mathcal{C}_j(S, x) = x^{i_j}$ for all $j \in
  \os{1, \dots, k}$ and $x \in S$. Using $k-1$ additional product gates, these
  circuits can be combined to a single circuit $\mathcal{C}$ with
  $\mathcal{C}(S, x_1, \dots, x_k) = x_1^{i_1} \cdots x_k^{i_k} = y$.

  In total, the resulting circuit consists of $k \cdot 2\ceil{\log N} + k - 1 <
  5 \log^2(N)$ gates.
\end{proof}

Let $G$ be a finite group and let $X$ be a subset of $G$. A sequence $(g_1,
\dots, g_\ell)$ of elements of $G$ is a \emph{straight-line program over $X$}
if for each $i \in \os{1, \dots, \ell}$, we have $g_i \in X$ or $g_i =
g_p^{-1}$ or $g_i = g_p g_q$ for some $p, q < i$.
The number $\ell$ is the \emph{length} of the straight-line program and the
elements of the sequence are said to be \emph{generated} by the straight-line
program.
The following result by Babai and Szemer{\'e}di~\cite{BabaiS84} is commonly
known as \emph{Reachability Lemma}.

\begin{lemma}[Reachability Lemma]
  Let $G$ be a finite group and let $X$ be a set of generators of $G$.
  Then, for each element $t \in G$, there exists a straight-line program over
  $X$ generating $t$ which has length at most $(\log \abs{G} + 1)^2$.
  \label{lem:reachability}
\end{lemma}

The proof of this lemma is based on a technique called ``cube doubling''. For
details, we refer to~\cite{Babai91}.
It is now easy to see that groups admit poly-logarithmic circuits.

\begin{lemma}
  The variety $\vG$ has the {\SPLCP}.
  \label{lem:groups}
\end{lemma}
\begin{proof}
  Let $G$ be a group of order $N$, let $X$ be a subset of $G$ and let $y$ be an
  element in the subgroup of $G$ generated by $X$.
  By Lemma~\ref{lem:reachability}, we know that there exists a straight-line
  program $(g_1, \dots, g_\ell)$ over $X$ with $\ell \le (\log(N) + 1)^2$ and
  $g_\ell = y$.
  We may assume that the elements $g_1, \dots, g_\ell$ are pairwise distinct.
  It suffices to describe how to convert this straight-line program into a
  Cayley circuit $\C$ and values $x_1, \dots, x_k \in X$ such that $\C(S, x_1,
  \dots, x_k) = y$.

  We start with an empty circuit and with $k = 0$ and process the elements of
  the straight-line program left to right. For each element $g_i$, we add gates
  to the circuit. The output gate of the circuit obtained after processing the
  element $g_i$ will be called the \emph{$g_i$-gate}.

  If the current element $g_i$ is contained in $X$, we increment $k$, add a new
  input gate to the circuit and let $x_k = g_i$.
  If the current element $g_i$ can be written as a product $g_p g_q$ with
  $p, q < i$, we add a new product gate to the circuit and connect the
  $g_p$-gate as well as the $g_q$-gate to this new gate.
  If the current element $g_i$ is an inverse $g_p^{-1}$ with $p < i$, we take a
  circuit $\C'$ with $2\ceil{\log N}$ gates and with $\C'(G, x) = x^{N-1}$ for
  all $x \in S$.
  Such a circuit can be built by using the powering technique illustrated in
  Example~\ref{ex:power}.
  We add $\C'$ to $\C$, replacing its input gate
  by an edge coming from the $g_p$-gate.

  The resulting circuit has size at most $(\log(N) + 1)^2 \cdot 2\ceil{\log N}
  \le 2(\log(N) + 1)^3$.
\end{proof}

We will now show that for classes of semigroups with the {\SPLCP}, one can
solve the {\SCSMP} in $\qAC^0$.

\begin{theorem}
  Let $\vV$ be a class of semigroups with the {\SPLCP}.
  Then $\CSM(\vV)$ is in $\qAC^0$.
  \label{thm:ACzero}
\end{theorem}
\begin{proof}
  We construct a family of unbounded fan-in constant-depth Boolean circuits
  with quasi-polynomial size, deciding, given the multiplication table of a
  semigroup $S \in \vV$, a set $X \subseteq S$ and an element $t \in S$ as
  inputs, whether $t$ is in the subsemigroup generated by $X$.

  Since $\vV$ has the {\SPLCP}, we know that, for some constant $c \in \N$, the
  element $t$ is in the subsemigroup generated by $X$ if and only if there
  exist a Cayley circuit $\C$ of size $\log^c(n)$ and inputs $x_1, \dots, x_k
  \in X$ such that $\mathcal{C}(S, x_1, \dots, x_k) = t$.
  There are at most $(\log^{c}(n) \cdot \log^{c}(n))^{\log^c(n)} = 2^{\log^c(n)
  \log(2c \log n)}$ different Cayley circuits of this size.
  Let us consider one of these Cayley circuits $\C$. Suppose that $\C$ has $k$
  input gates.
  By Proposition~\ref{prop:simulation}, there exists a unbounded fan-in
  constant-depth Boolean circuit of size $n^{\log^c n} = 2^{\log^{c+1} n}$
  deciding on input $S$ and elements $x_1, \dots, x_k \in S$ whether $\C(S,
  x_1, \dots, x_k) = t$.
  There are at most $n^k \le n^{\log^c n} = 2^{\log^{c+1} n}$ possibilities of
  connecting (not necessarily all) input gates corresponding to the elements of
  $X$ to this simulation circuit.

  Thus, we can check for all Cayley circuits of the given size and all possible
  input assignments in parallel, whether the value of the corresponding circuit
  is $t$, and feed the results of all these checks into a single OR gate to
  obtain a quasi-polynomial-size Boolean circuit.
\end{proof}

In conjunction with Lemma~\ref{lem:com} and Lemma~\ref{lem:groups}, we
immediately obtain the following corollary.

\begin{corollary}
  Both $\CSM(\vG)$ and $\CSM(\vCom)$ are contained in $\qAC^0$.
  \label{crl:ACzero}
\end{corollary}

As stated in the preliminaries, problems in $\qAC^0$ cannot be hard for any
complexity class containing $\Parity$. Thus, we also obtain the following
statement.

\begin{corollary}
  Let $\vV$ be a class of semigroups with the {\SPLCP}, such as the variety of
  finite groups $\vG$ or the variety of finite commutative semigroups $\vCom$.
  Then $\CSM(\vV)$ is not hard for any complexity class containing $\Parity$,
  such as $\ACC^0$, $\TC^0$, $\NC^1$, $\LOGSPACE$ or $\NL$.
  \label{crl:parity}
\end{corollary}

\subsection{The Complexity Landscape of {\CCSM}}

Our hardness results and $\qAC^0$-algorithms have an immediate consequence on
algebraic properties of maximal classes of finite semigroups for which the
{\SCSMP} can be decided in $\qAC^0$.
It relies on the following result, which can be seen as a consequence
of~\cite{Almeida88} and the fact that the zero element in a semigroup is always
central.
For completeness, we provide a short and self-contained proof.

\begin{proposition}
  The variety $\vNil$ is included in $\vG \lor \vCom$.
\end{proposition}

\begin{proof}
  We show that every finite nilpotent semigroup is a quotient of a subdirect
  product of a finite group and a finite commutative semigroup.
  Note that in a finite nilpotent semigroup~$S$, there exists an integer $e \ge
  0$ such that for each $x \in S$, the power $x^e$ is the zero element.
  Let $T = \os{1, \dots, e}$ be the commutative semigroup with the product of
  two elements $i$ and $j$ defined as $\min\os{i+j,e}$.

  Let $G$ be a finite group generated by the set $X$ of non-zero elements of
  $S$ such that no two products of less than $e$ elements of $X$ evaluate to
  the same element of $G$. Such a group exists because the free group over $X$
  is residually finite~\cite{Levi33}.

  Let $U$ be the subsemigroup of $G \times T$ generated by $\set{(x, 1)}{x \in
  X}$.
  Now, we define a mapping $\varphi \colon U \to S$ as follows.
  Each element of the form $(g, e)$ is mapped to zero.
  For every $(g, \ell)$ with $\ell < e$, there exists, by choice of $G$ and by
  the definition of $U$, a unique factorization $g = x_1 \cdots x_\ell$ with
  $x_1, \dots, x_\ell \in X$.  We map $(g, \ell)$ to the product $x_1 \cdots
  x_\ell$ evaluated in $S$.
  It is straightforward to verify that $\varphi$ is a surjective morphism and
  thus, $S$ is a quotient of $U$.
\end{proof}

\begin{corollary}
  There exist two varieties $\vV$ and $\vW$ such that both $\CSM(\vV)$ and
  $\CSM(\vW)$ are contained in $\qAC^0$ (and thus not hard for any class
  containing $\Parity$) but $\CSM(\vV \lor \vW)$ is $\NL$-complete.
\end{corollary}

The corollary is a direct consequence of the previous proposition,
Corollary~\ref{crl:ACzero} and Theorem~\ref{thm:nil}.
As was observed in~\cite{BarringtonKLM01} already, Cayley semigroup problems
seem to have ``strange complexity''.
The previous result makes this intuition more concrete and suggests that it is
difficult to find ``nice'' descriptions of maximal classes of semigroups for
which the {\SCSMP} is easier than any $\NL$-complete problem.

\subsection{Connections to \texorpdfstring{$\FOLL$}{FOLL}}
\label{sec:foll}

In a first attempt to solve outstanding complexity questions related to the
{\SCSMP}, Barrington {et al.} introduced the complexity class $\FOLL$. The
approach presented in the present paper is quite different. This raises the
question of whether our techniques can be used to design $\FOLL$-algorithms for
{\SCSM}.
Note that $\FOLL$ and $\qAC^0$ are known to be incomparable, so we cannot
use generic results from complexity theory to simulate $\qAC^0$ circuits using
families of $\FOLL$ circuits or vice versa.
The direction $\FOLL \not\subseteq \qAC^0$ follows from bounds on the average
sensitivity of bounded-depth circuits~\cite{Boppana97}; using these bounds, one
can show that there exists a padded version of the $\Parity$ function which can
be computed by a $\FOLL$ circuit family and cannot be computed by any $\qAC^0$
circuit family.
Conversely, each subset of $\os{0, 1}^n$ of cardinality at most $n^{\log n}$ is
decidable by a depth-$2$ circuit of size $n^{1+\log n} + 1$, but for each fixed
$k \in \N$, there is some large value $n \ge 1$ such that the number of such
subsets exceeds the number of different circuits of size $n^k$.
This shows that there exist languages in $\qAC^0$ which are not contained in
$\Ppoly \supseteq \FOLL$.

Designing an $\FOLL$-algorithm which works for arbitrary classes of semigroups
with the {\SPLCP} seems difficult.
However, for certain special cases, there is an interesting approach, based on
the repeated squaring technique. In the remainder of this section, we sketch
one such special case.

For a Cayley circuit, the \emph{width} of a topological ordering $(v_1, \ldots,
v_m)$ of the gates is the smallest number $w \in \N$ such that for each $i \in
\os{1, \ldots, m-1}$, at most $w$ product gates from the set $A_i = \os{v_1,
\ldots, v_i}$ are connected to gates in $B_i = \os{v_{i+1}, \ldots, v_m}$.
Let $C_i$ be the set of product gates, which belong to $A_i$ and are connected
to gates in $B_i$.
The subcircuit induced by $A_i$ can be interpreted as a Cayley circuit
computing multiple output values $C_i$.
The subcircuit induced by $B_i$ can be seen as a circuit which, in addition to
the input gates of the original circuit, uses the gates from $C_i$ as input
gates.
The \emph{width} of a Cayley circuit is the smallest width of a topological
ordering of its gates. Let us fix some width $w \in \N$.

We introduce a predicate $P(z_1, \dots, z_w, y_1, \dots, y_w, i)$ which is true
if there exists a Cayley circuit of width at most $w$ and size at most $2^i$
with $w$ additional input gates and $w$ additional \emph{passthrough gates}
(which have in-degree 1 and replicate the value of their predecessors), such
that the elements $y_1, \dots, y_w \in S$ occur as values of the passthrough
gates when using $z_1, \dots, z_w \in S$ as values for the additional input
gates and using any subset of the original inputs $X$ as values for the
remaining input gates.
The additional input gates (resp.~passthrough gates) are not counted when
measuring the circuit size but are considered as product gates when measuring
width and they have to be the first (resp.~last) gates in all topological
orderings considered for width measurement.
For each fixed~$i$, there are only $n^{2w}$ such predicates.

The truth value of a predicate with $i = 0$ can be computed by a constant-depth
unbounded fan-in Boolean circuit of polynomial size. This is achieved by
computing all binary products of the elements $z_1, \dots, z_w$ and elements of
the input set $X$.
For $i \ge 1$, the predicate $P(z_1, \dots, z_w, y_1, \dots, y_w, i)$ is true if
and only if there exist $z_1', \dots, z_w' \in S$ such that both $P(z_1, \dots,
z_w, z_1', \dots, z_w', i-1)$ and $P(z_1', \dots, z_w', y_1, \dots, y_w, i-1)$
are true.
Having the truth values of all tuples for $i-1$ at hand, this can be checked
with a polynomial number of gates in constant depth because there are only
$n^w$ different vectors $(z_1', \dots, z_w') \in S^w$.

For a class of semigroups with Cayley circuits of bounded width and
poly-logarithmic size, we obtain a circuit family of depth $\bigO(\log \log n)$
deciding {\SCSM}:
the predicates are computed for increasing values of $i$, until $i$ exceeds the
logarithm of an upper bound for the Cayley circuit size and then, we return
$P(x, \dots, x, t, \dots, t, i)$ for the element $t$ given in the input and for
an arbitrary element $x \in X$.
It is worth noting that the circuits constructed in the proof of
Proposition~\ref{prop:lpbp-plrp} have width at most $2$, so our
$\FOLL$-algorithm is a generalization of the \emph{Double-Barrelled Recursive
Strategy} and the proof that $\CSM(\vAb) \in \FOLL$ presented
in~\cite{BarringtonKLM01}.
In particular, the procedure above yields a self-contained proof of the
following result.

\begin{theorem}
  Let $\vV$ be a class of semigroups which is closed under taking subsemigroups
  and has the {\SLPBP}.
  Then $\CSM(\vV)$ is in $\FOLL$.
  \label{thm:FOLL}
\end{theorem}

By Lemma~\ref{lem:com}, we obtain the following corollary.

\begin{corollary}
  $\CSM(\vCom)$ is contained in $\FOLL$.
  \label{thm:Com-FOLL}
\end{corollary}

\section{Summary and Outlook}

We provided new insights into the complexity of the {\SCSMP} for classes of
finite semigroups, giving parallel algorithms for the variety of finite
commutative semigroups and the variety of finite groups. We also showed that a
maximal class of semigroups with {\SCSM} decidable by $\qAC^0$ circuits does
not form a variety. Afterwards, we discussed applicability to $\FOLL$.

It is tempting to ask whether one can find nice connections between algebra and
the complexity of the {\SCSMP} by conducting a more fine-grained analysis.
For example, it is easy to see that for the varieties of \emph{rectangular
bands} and \emph{semilattices}, the {\SCSMP} is in $\AC^0$.
Does the maximal class of finite semigroups, for which the {\SCSMP} is in
$\AC^0$, form a variety of finite semigroups?
Is it possible to show that $\AC^0$ does not contain $\CSM(\vG)$?
Potential approaches to tackling the latter question are reducing small
distance connectivity for paths of non-constant length~\cite{ChenOSRT16} to
$\CSM(\vG)$ or developing a suitable switching lemma.
Another related question is whether there exist classes of semigroups for which
the {\SCSMP} cannot be $\NL$-hard but, at the same time, is not contained
within~$\qAC^0$.

Moreover, it would be interesting to see whether the {\SCSMP} can be shown to
be in $\FOLL$ for all classes of semigroups with the {\SPLCP}.
More generally, investigating the relation between $\FOLL$ and $\qAC^0$, as
well as their relationships to other complexity classes, remains an interesting
subject for future research.

\subparagraph{Acknowledgements.} I would like to thank Armin Weiß for several
interesting and inspiring discussions, and for pointing out that $\qAC^0$ is
not contained within $\Ppoly$.
I would also like to thank Samuel Schlesinger for comments which led to an
improved presentation of the proof of Proposition~\ref{prop:simulation} and for
pointing out that results on the average sensitivity of bounded-depth circuits
can be used to show that $\FOLL$ is not contained within $\qAC^0$.
Moreover, I am grateful to the anonymous referees for providing helpful
comments that improved the paper.

\newcommand{\Ju}{Ju}\newcommand{\Ph}{Ph}\newcommand{\Th}{Th}\newcommand{\Ch}{Ch}\newcommand{\Yu}{Yu}\newcommand{\Zh}{Zh}\newcommand{\St}{St}\newcommand{\curlybraces}[1]{\{#1\}}

\end{document}